\documentclass[11pt, a4paper]{amsart}
\usepackage{amsfonts}

\usepackage[dvips]{epsfig}
\usepackage{amsmath}
\usepackage{amsthm}
\usepackage{amsbsy}
\usepackage{amsgen}
\usepackage{amscd}
\usepackage{amsopn}
\usepackage{amstext}
\usepackage{amsxtra}
\usepackage{dsfont}

\newcommand{\var}{\operatorname{Var}}

\newtheorem{theorem}{Theorem}[section]
\newtheorem{proposition}[theorem]{Proposition}
\newtheorem{lemma}[theorem]{Lemma}

\newtheorem{conjecture}[theorem]{Conjecture}

\newtheorem{remark}[theorem]{Remark}

\begin{document}
\title{The defect variance of random spherical harmonics}
\author{Domenico Marinucci and Igor Wigman}
\address{Department of Mathematics, University of Rome Tor Vergata}
\email{marinucc@mat.uniroma2.it}
\address{Department of Mathematics, Cardiff University, Wales, UK}
\email{wigmani@cardiff.ac.uk}
\thanks{IW was supported by the Knut and Alice Wallenberg Foundation, grant
KAW.2005.0098}
\maketitle

\begin{abstract}
The defect of a function $f:M\rightarrow \mathbb{R}$ is defined as the
difference between the measure of the positive and negative regions. In this
paper, we begin the analysis of the distribution of defect of random
Gaussian spherical harmonics. By an easy argument, the defect is non-trivial
only for even degree and the expected value always vanishes. Our principal
result is obtaining the asymptotic shape of the defect variance, in the high
frequency limit. As other geometric functionals of random eigenfunctions,
the defect may be used as a tool to probe the statistical properties of
spherical random fields, a topic of great interest for modern Cosmological
data analysis.
\end{abstract}

\section{Introduction}

In recent years, a lot of interest has been drawn by the analysis of the
geometric features of random eigenfunctions for the spherical Laplacian.
More precisely, let us denote as usual by $\left\{ Y_{lm}(.)\right\}
_{m=-l,...,l},$ for $l=1,2,...$ the set of (real-valued) spherical
harmonics, i.e. the array of real-valued functions satisfying
\begin{equation*}
\Delta _{S^{2}}Y_{lm}=-l(l+1)Y_{lm}\text{,}
\end{equation*}
where $\Delta _{S^{2}}$ denotes the spherical Laplacian (see for instance
\cite{VMK}, \cite{Wig2}). The random model we shall focus on is
\begin{equation}
f_{l}(x)=\frac{1}{\sqrt{2l+1}}\sum_{m=-l}^{l}a_{lm}Y_{lm}(x)
\label{eq:fl def}
\end{equation}%
where the coefficients $\left\{ a_{lm}\right\} $ are independent standard
Gaussian with zero-mean and unit variance.

The random field $f_{l}$ is isotropic, meaning that for any $k\in \mathbb{N}$
and $x_{1},\ldots ,x_{k}\in \mathcal{S}^{2}$, the distribution of the random
vector $\left\{ f_{l}(x_{1}),\ldots ,f_{l}(x_{l})\right\} $ is invariant
under rotations, i.e. equals to the distribution of $\left\{ f_{l}(g\cdot
x_{1}),\ldots ,f_{l}(g\cdot x_{l})\right\} $ for any rotation $g\in SO(3)$
on the sphere. Also, $f_{l}$ is centred Gaussian, with covariance function
\begin{equation}
r_{l}(x,y):=\mathbb{E}[f_{l}(x)\cdot f_{l}(y)]=P_{l}(\cos (d(x,y))),
\label{eq:rl coval def}
\end{equation}
where $P_{l}$ are the usual Legendre polynomials defined by Rodrigues'
formula
\begin{equation*}
P_{l}(t):=\frac{1}{2^{l}l!}\frac{d^{l}}{dt^{l}}(t^{2}-1)^{l}
\end{equation*}
and $d(x,y)$ is the (spherical) geodesic distance between $x$ and $y$. As
well-known, Legendre polynomials are orthogonal w.r.t. the constant weight $
\omega (t)\equiv 1$ on $[-1,1]$, see Section \ref{sec:Legendre moments} for
more discussion and details.

\subsection{Background}

Random eigenfunctions for the spherical Laplacian naturally emerge in a
number of different physical contexts. A particularly active area is related
to the analysis of isotropic spherical random fields on the sphere, as
motivated for instance by the analysis of Cosmic Microwave Background
radiation (CMB), see for instance \cite{dodelson,durrer} and the references
therein. Under these circumstances, eigenfunctions like $f_{l}$ represent
the (normalized) Fourier components of the field, i.e. the following
orthogonal expansion holds, in the $L^2$ sense:
\begin{equation*}
f(x)=\sum_{l=1}^{\infty }c_{l}f_{l}(x)\text{ ,}
\end{equation*}
where $c_{l}\geq 0$ is a deterministic sequence (depending on the angular
power spectrum) which encodes the full correlation structure of $f.$

In the CMB literature, $f(x)$, and hence the components $f_{l}$, are
actually observed by highly sophisticated satellite experiments such as
\emph{WMAP} by NASA and \emph{Planck} by ESA. It is then common practice to
analyze geometric functionals of the observed CMB radiation to constraint
the statistical properties of the underlying fields, e.g. to test for
isotropy and/or Gaussianity. For instance, the three Minkowski functionals,
providing the area, the boundary length and the topological genus of
excursion sets over a given level $\left( A_{z}:=\left\{x\in \mathcal{S}%
^{2}:f(x)\geq z\right\} \right)$ have been applied on CMB data by a huge
number of authors, including \cite{hikage1,hikage2,matsubara}. Many efforts
have been spent to analyze unexpected features; for instance, the so-called
Cold Spot (see \cite{Cruz1,Cruz2} and the references therein), whose
statistical significance can also be evaluated by means of local curvature
properties of isotropic Gaussian fields (\cite{hansen}).

In this work, we shall focus on one of the most important geometric
functionals, namely the \emph{defect}. The defect (or ``signed area'', see 
\cite{BGS}) of a function $\psi :\mathbb{S}^{2}\rightarrow \mathbb{R}$ is
defined as
\begin{equation}
\mathcal{D}(\psi ):=\mathrm{meas}\left( \psi ^{-1}(0,\infty )\right) -
\mathrm{\ meas}\left( \psi ^{-1}(-\infty ,0)\right) =\int_{\mathcal{S}^{2}}
\mathcal{H}(\psi (x))dx.  \label{eq:defect def}
\end{equation}
Here $\mathcal{H}(t)$ is the Heaviside function
\begin{equation}
\mathcal{H}(t)=\mathds{1}_{[0,\infty) }(t)-\mathds{1}_{(-\infty ,0]}(t)=
\begin{cases}
1 & t>0 \\
-1 & t<0 \\
0 & t=0
\end{cases}
,  \label{eq:heaviside def}
\end{equation}
where $\mathds{1}_{A}(t)$ is the usual indicator function of the set $A$,
and $dx$ is the Lebesgue measure. The defect is hence the difference between the
areas of positive and negative inverse image of $\psi $, respectively.

Of course, the defect $\mathcal{D}_{l}=\mathcal{D}(f_{l})$ of $f_{l}$ is a
random variable; its distribution is the primary focus of the present paper.
Note that for odd $l$, $f_{l}$ is always odd, so that in this case the
defect vanishes identically ($\mathcal{D}_{l}\equiv 0$), and therefore
$\mathcal{D}_{l}$ has nontrivial distribution for even $l$ only. A possible
alternative to avoid trivialities is to restrict ourselves to a subset
of the sphere, the most natural choice being a hemisphere, i.e. we may
choose any hemisphere $\mathcal{E}\subseteq \mathcal{S}^{2}$ and define
\begin{equation*}
\mathcal{D}_{l}^{\mathcal{E}}:=\int_{\mathcal{E}}\mathcal{H}(f_{l}(x))dx;
\end{equation*}
in this paper we study $\mathcal{D}_{l}$ for even $l$ only.

\subsection{Statement of the main result}

\label{sec:main res}

Evaluating the expectation of $\mathcal{D}_{l}$ is trivial. Indeed, notice
that integration over $\mathcal{S}^{2}$ is exchangeable with expectation, so
that
\begin{equation*}
\mathbb{E}[\mathcal{D}_{l}]=\int_{\mathcal{S}^{2}}\mathbb{E}\left[ \mathcal{H%
}(f_{l}(x))\right] dx,
\end{equation*}%
and $\mathbb{E}\left[ \mathcal{H}(f_{l}(x))\right] =0$ vanishes for every $%
x\in \mathcal{S}^{2}$, by the symmetry of the Gaussian distribution. We just
established the following lemma:

\begin{lemma}
For every $l=1,2,...,$ we have
\begin{equation*}
\mathbb{E}[\mathcal{D}_{l}]=0.
\end{equation*}
\end{lemma}

The main result of the present paper concerns the asymptotic behaviour of
the defect variance:

\begin{theorem}
\label{thm:var=1/l^2} As $l\rightarrow \infty $ along even integers, the
defect variance is asymptotic to
\begin{equation}
\var(\mathcal{D}_{l})=\frac{C}{l^{2}}(1+o(1)),  \label{eq:var=1/l^2}
\end{equation}
where $C>0$ is a positive constant.
\end{theorem}

The constant $C$ in \eqref{eq:var=1/l^2} may be expressed in terms of the
infinite (conditionally convergent) integral
\begin{equation}
C=32\pi \int\limits_{0}^{\infty }\psi \left( \arcsin (J_{0}(\psi
))-J_{0}(\psi )\right) d\psi .  \label{eq:C=*int}
\end{equation}%
See Section \ref{sec:on the proof} for some details on the integral on the
RHS of \eqref{eq:C=*int}. We do not know whether one can evaluate $C$
explicitly; however, we shall be able to show that
\begin{equation*}
C>\frac{32}{\sqrt{27}}\text{ ,}
\end{equation*}
(see Lemma \ref{lem:C>2/pi^2 /sqrt(27)}).

One should compare the statement of Theorem \ref{thm:var=1/l^2} to
empirical results of the study conducted by Blum, Gnutzmann and Smilansky ~\cite{BGS}.
The authors of that work studied the defect (or, as they refer to, the ``signed area")
of random monochromatic waves on various planar domains. For the particular
case of unit circle, they found that the order of magnitude of defect
variance is consistent to \eqref{eq:var=1/l^2} per unit area, with leading constant evaluated numerically
as $\approx 0.0386$. The corresponding defect variance per unit area in our situation is
\begin{equation*}
\var\left(\frac{\mathcal{D}_{l}}{4\pi } \right)=\frac{\widetilde{C}}{l^{2}}
(1+o(1))
\end{equation*}
with leading constant $$\widetilde{C}>\frac{2}{\pi ^{2}\times \sqrt{27}} =
0.0389...$$

\subsection{Previous work}

\label{sec:prev work}

To put our results in a proper perspective and explain the technical
difficulties to be handled, we need to briefly recall some results from ~\cite{MaWi}.
In that work we studied the asymptotic behaviour of excursion
sets $$A_{z} = f_{l}^{-1}\left([z,\infty)\right)$$ of $f_{l}$, as $l\rightarrow\infty$. As we will recall below, a
rather peculiar phenomenon can be shown to hold for $z\neq 0; $ namely, the
asymptotic distribution of the area of these excursion sets is fully
degenerate over $z$, i.e. it corresponds to a Gaussian random variable times
a deterministic function of $z$. This result is intuitively due to an
asymptotic degeneracy in the excursion set functionals, which turns out to
be dominated by a single polynomial (quadratic) term.

This degeneracy, however, does not hold for the special case $z=0.$ More
precisely, in \cite{MaWi} we focused on the asymptotic behaviour of the
empirical measure of random spherical harmonics, defined as
\begin{equation*}
\Phi _{l}(z):=\int_{S^{2}}\mathds{1}_{[z,\infty )}(f_{l}(x))dx=\mathrm{meas}
\left\{ x:f_{l}(x)\leq z\right\},
\end{equation*}
$z\in \mathbb{R}$.
A key step in that paper is the asymptotic expansion
\begin{equation}
\Phi _{l}(z)=4\pi \times \Phi (z)+\sum_{q=1}^{\infty }\frac{J_{q}(z)}{q!}%
h_{l;q},\text{ }h_{l;q}:=\int_{\mathbb{S}^{2}}H_{q}(f_{l}(x))dx\text{,}
\label{asexp2}
\end{equation}%
where $H_{q}(.)$ are standard Hermite polynomials, $\Phi (z)=\Pr \left\{
Z\leq z\right\} $ is the cumulative distribution function of a standard
Gaussian variable, and the deterministic functions $J_{q}(z)$ can be
explicitly provided in terms of higher order derivatives of $\Phi (z),$ $%
J_{q}(z)=(-1)^{q}\Phi ^{(q)}(z)$ (see \cite{MaWi} for more discussion and
details).

It turns out that, as $l\rightarrow \infty $,
\begin{equation*}
\frac{\sqrt{l}}{4\pi}h_{l;2}\rightarrow _{d}N(0,1)\text{ , }\sqrt{l}h_{l;q}=o_{p}(1)\text{
, for }q\geq 3\text{;}
\end{equation*}%
moreover $J_{2}(z)=-z\phi (z)$ clearly does not vanish for all $z\neq
0,$ whence the asymptotic behaviour of $\sqrt{l}\left\{ \Phi _{l}(z)-4\pi
\times \Phi (z)\right\} $ is easily seen to be Gaussian, uniformly over $z$.
Furthermore, the limiting process is completely degenerate with respect to $z,$
a feature to which we shall come back later.

For $z\neq 0,$ the asymptotic behaviour of the area functional for the
excursion sets is hence fully understood. The previous argument, however,
fails for $z=0,$ as in this case the leading term is null and each summand
in the asymptotic expansion (\ref{asexp2}) becomes relevant. Up to a linear
transformation, this is clearly equivalent to the defect functional, indeed%
\begin{equation*}
\mathcal{D}_{l}=4\pi -2\Phi _{l}(0)\text{ .}
\end{equation*}%
Thus, the case $z=0$ is the most challenging from the mathematical point of
view, and, at the same time, the most interesting from the point of view of
geometric interpretation.

\subsection{Overview of the paper}

The plan of the paper is as follows. In Section \ref{sec:on the proof}, we
provide the main ideas behind our principal arguments, to help the reader
understand the material to follow; in Section \ref{sec:discussion}, we
discuss the relation of our results to recent works on the
distribution of nodal lengths and level curves for random eigenfunctions and
related conjectures; Section \ref{proof} provides the proof of the main
results, whereas Section \ref{sec:Legendre moments} contains auxiliary lemmas,
we believe of independent interest, on the asymptotic behaviour of
moments of Legendre polynomials.

\subsection{Acknowledgements}

We wish to thank Ze\'{e}v Rudnick for suggesting the problem, many
stimulating and fruitful discussions and useful comments on an earlier
version of this paper, and Mikhail Sodin for many
stimulating and fruitful discussions. A substantial part of this research
was done during the second author's visit to University of Rome ``Tor
Vergata", and he would like to acknowledge the extremely friendly and
stimulating environment in the institution, and the generous financial
support.

\section{On the proof of Theorem \ref{thm:var=1/l^2}}

\label{sec:on the proof}

To establish our results, we shall need a detailed analysis of the odd
moments of Legendre polynomials
\begin{equation*}
\int\limits_{0}^{1}P_{l}(t )^{2k+1}dt = \int\limits_{0}^{\pi /2}P_{l}(\cos
\theta )^{2k+1}\sin {\theta }d\theta .
\end{equation*}%
The rationale for this can be explained as follows. It is relatively easy to
express (up to a constant) the defect variance as
\begin{equation}
I_{l}=\int\limits_{0}^{\pi /2}\arcsin (P_{l}(\cos \theta ))\sin {\theta }%
d\theta ,  \label{eq:Il def scal}
\end{equation}%
where $\sin {\theta }d\theta $ is (up to a constant) the uniform measure on
the sphere in the spherical coordinates (see Lemma \ref%
{lem:var(Dl)=*int(arcsin)}). It then remains to understand the asymptotic
behaviour of $I_{l}$, a task which is put forward in Proposition \ref%
{prop:Il sim C1/l^2}; here we provide the main ideas underlying its proof.
We know from Hilb's asymptotics (Lemma \ref{lem:Hilb}), that $P_{l}(\cos
\theta )$ has a scaling limit: as $l\rightarrow \infty $, for any \emph{fixed%
}
\begin{equation*}
\psi \in \left[ 0,(l+1/2)\frac{\pi }{2}\right] ,
\end{equation*}%
we have
\begin{equation}
P_{l}\left( \cos \left( \frac{\psi }{l+1/2}\right) \right) \approx
J_{0}(\psi ).  \label{eq:Pl scal}
\end{equation}%
In fact, the latter estimate holds uniformly for $\psi =o(l)$, and it may be
shown that, as a tail of a conditionally convergent integral, the
contribution of the other regime $\psi \gg l$ is negligible (see the proof
of Proposition \ref{prop:Il sim C1/l^2}).

In this sequel we neglect the difference between $l$ and $l+\frac{1}{2}$. It
is then natural to try to replace $P_{l}(\cos \theta )$ in
\eqref{eq:Il def
scal} with its scaling limit; a formal substitution yields heuristically
\begin{equation}
\begin{split}
I_{l}& \approx \frac{1}{l}\int\limits_{0}^{l\pi /2}\arcsin \left(
P_{l}\left( \cos \left( \psi /l\right) \right) \right) \sin \left( \psi
/l\right) d\psi \approx \frac{1}{l}\int\limits_{0}^{l\pi /2}\arcsin
(J_{0}(\psi ))\sin \left( \psi /l\right) d\psi \\
& \approx \frac{1}{l^{2}}\int\limits_{0}^{l\pi /2}\arcsin (J_{0}(\psi ))\psi
d\psi ,
\end{split}
\label{eq:Il approx heurist}
\end{equation}%
where we replaced $\sin \left( \frac{\psi }{l}\right) $ with $\frac{\psi }{l}
$, which is justified for $\psi =o(l)$. This is inconsistent to the
statement of Theorem \ref{thm:var=1/l^2}; the integral
\begin{equation}
\int\limits_{0}^{\infty }\arcsin (J_{0}(\psi ))\psi d\psi ,
\label{eq:int J(psi)*psi}
\end{equation}%
diverges, so some more care is needed to transform \eqref{eq:Il approx
heurist}. Moreover, the last integrand in \eqref{eq:Il approx heurist} is
somewhat different from the integrand in the integral \eqref{eq:C=*int}
defining $C$ (neglecting the constant in front of the integral in %
\eqref{eq:C=*int}) in that in \eqref{eq:C=*int} we subtract $J_{0}(\psi )$
from the arcsine. Note that the factor $\psi $ in both
\eqref{eq:Il approx
heurist} and \eqref{eq:C=*int} is reminiscent of the uniform measure $\sin
\theta d\theta $ on the sphere.

There is a subtlety that explains the latter discrepancy between the
integrand in the definition \eqref{eq:C=*int} of $C$ and the integrand in %
\eqref{eq:Il approx heurist}. One way to justify the substitution of the
scaling limit of $P_{l}(\cos \theta )$ (i.e. the second step of
\eqref{eq:Il
approx heurist}) is expanding the arcsine in the integral of
\eqref{eq:Il
def scal} into the Taylor series around the origin
\begin{equation*}
\arcsin (t)=\sum\limits_{k=0}^{\infty }a_{k}t^{2k+1},
\end{equation*}%
for some explicitly given coefficients $a_{k}$ (see
\eqref{eq:arcsin(t)-t
Tay} and the formula immediately after). We then need to evaluate all the
odd moments of $P_{l}(\cos \theta )$ w.r.t. to the measure $\sin \theta
d\theta $ on $\left[ 0,\frac{\pi }{2}\right] $. It turns out that the
moments (appropriately scaled) are asymptotically equivalent to the
corresponding moments of the Bessel function on $\mathbb{R}_{+}$ w.r.t. the
measure $\psi d\psi $, with the exception for the first moment (Lemma \ref%
{lem:int Pl^k =ck/l^2}). Indeed, $P_{l}(\cos \theta )$ integrates to zero,
whereas $\int\limits_{0}^{\infty }J_{0}(\psi )\psi d\psi $ diverges; this
phenomenon also accounts for the divergence of \eqref{eq:int
J(psi)*psi}. To account for this difference, we need to subtract $J_{0}(\psi
)$ from the arcsine in the integrand \eqref{eq:Il approx heurist}, thus
obtaining a conditionally convergent integral. This is indeed the heuristic
explanation for the discrepancy between \eqref{eq:Il approx heurist} and %
\eqref{eq:C=*int}.

To summarize this discussion, we note the following. The variance is given,
up to an explicit constant, by the integral \eqref{eq:Il def scal}. The
constant $C$ in \eqref{eq:C=*int} is reminiscent of $I_{l}$ up to the
scaling \eqref{eq:Pl scal}: the expression $P_{l}(\cos \theta )$ is replaced
by its scaling limit $J_{0}(\psi )$, and the measure $\sin \theta d\theta $
becomes $\psi d\psi $ on $\mathbb{R}_{+}$ after scaling. Subtracting $%
J_{0}(\psi )$ from the integrand in \eqref{eq:C=*int} accounts for the fact
that for even $l$, the integral of $P_{l}(\cos \theta )$ against $\sin
\theta d\theta $ vanishes, whereas this is not the case of integral of the
scaling limit $J_{0}(\psi )$ against $\psi d\psi $.

\section{Discussion}

\label{sec:discussion}

As reported in Section \ref{sec:prev work}, the asymptotics for
the variance of the excursion sets is given for $z\neq 0$ by
\begin{equation*}
\var(\Phi _{l}(z))\sim z^{2}\phi (z)^{2}\cdot \frac{4\pi^2}{l}
\end{equation*}%
(here $\phi $ is the Gaussian probability density); for the defect ($z=0$),
this gives only an upper bound $o\left( \frac{1}{l}\right) $. The principal
result of the present paper (Theorem \ref{thm:var=1/l^2}) states that the
defect variance \eqref{eq:var=1/l^2} is of order of magnitude $l^{-2}$. Our
explanation for this discrepancy is the disappearance, for $z=0$, of
quadratic term in the Hermite expansion of the function $\mathds{1}
(f_{l}\leq z),$ see ~\cite{MaWi}.

We should compare this situation to the length distribution of the level
curves. For $t\in \mathbb{R}$ let
\begin{equation*}
\mathcal{L}_{l}^{t}=\mathcal{L}^{t}(f_{l})=\mathrm{length}\left(
f_{l}^{-1}(t)\right) .
\end{equation*}%
be the (random) length of the level curve $f_{l}^{-1}(t)$, the most
important case being that of $t=0$; the corresponding level curve is called
the \emph{nodal line}. It is known ~\cite{Wig3} that for $t\in \mathbb{R}$,
the expected length is
\begin{equation*}
\mathbb{E}\left[ \mathcal{L}_{l}^{t}\right] \sim c_{1}e^{-t^{2}/2}\cdot l.
\end{equation*}
The variance is asymptotic to
\begin{equation*}
\var\left( \mathcal{L}_{l}^{t}\right) \sim c_{2}e^{-t^{2}}t^{4}\cdot l
\end{equation*}
for $t\neq 0$, whereas for the nodal length
\begin{equation*}
\var\left( \mathcal{L}_{l}^{0}\right) \sim c_{3}\log {l}.
\end{equation*}
It is therefore natural to conjecture that the discrepancy of the defect and
the empirical measure are somehow related to the discrepancy in the variance
of the nodal length. However, no simple explanation, as in the case of the
defect, is known to explain the discrepancy between the nodal length and the
non-vanishing level curves length.

Note that we may relate $\mathcal{L}_{l}^{t}$ to the excursion sets by the
identity
\begin{equation}
\Phi _{l}(z)=\int\limits_{-\infty }^{z}\mathcal{L}_{l}^{t}dt.
\label{eq:Phil=int lev len}
\end{equation}%
Let us recall that for two random variables $X$, $Y$, the correlation is
defined as
\begin{equation*}
\mathrm{Corr}(X,Y)=\frac{Cov(X,Y)}{\sqrt{\var(X)}\sqrt{\var(Y)}}.
\end{equation*}%
It is known ~\cite{Wig3} that for every $t_{1},t_{2}\in \mathbb{R}$, $%
\mathcal{L}_{l}^{t}$ become asymptotically fully dependent for large $l$, in
the sense that\footnotemark
\begin{equation}
\mathrm{Corr}\left( \mathcal{L}_{l}^{t_{1}},\mathcal{L}_{l}^{t_{2}}\right)
=1-o_{l\rightarrow \infty }(1).  \label{eq:Lt1 corr Lt2}
\end{equation}%
In fact, a much stronger statement regarding the rate of convergence was
proven. It then seems, that the asymptotic degeneracy of $\Phi _{l}(z)$ is,
essentially, an artifact of \eqref{eq:Phil=int lev len} and the asymptotic
full dependence of the level lengths \eqref{eq:Lt1 corr Lt2}.

\footnotetext{%
The asymptotic full dependence \eqref{eq:Lt1 corr Lt2} was proven under the
technical assumption $t_{1},t_{2}\neq 0$; it is natural to conjecture,
though, that a slight modification of the same argument will work if either
of $t_{i}$ vanishes.}

A possible explanation for the phenomenon \eqref{eq:Lt1 corr Lt2} is the
following conjecture due to Mikhail Sodin. Let $x\in \mathcal{S}^{2}$ and
for $t\in \mathbb{R}$ let $\mathcal{L}_{l;x}^{t}=\mathcal{L}_{x}^{t}(f_{l})$
(the ``local length'') be the (random) length of the unique component of $%
f_{l}^{-1}(t)$ that contains $x$ inside (or $0$, if $f_{l}$ does not cross
the level $t$).

\begin{conjecture}[M. Sodin]
\label{conj:loc len full dep} The local lengths are asymptotically fully
dependent in the sense that for every $x\in \mathcal{S}^{2}$ and $%
t_{1},t_{2}\in \mathbb{R}$
\begin{equation*}
\mathrm{Corr}\left( \mathcal{L}_{l;x}^{t_{1}},\mathcal{L}_{l;x}^{t_{2}}%
\right) =1-o_{l\rightarrow \infty }(1).
\end{equation*}
\end{conjecture}

Heuristically, Conjecture \ref{conj:loc len full dep} is stronger than %
\eqref{eq:Lt1 corr Lt2} (and thus the asymptotic degeneracy of the empirical
measure via \eqref{eq:Phil=int lev len}), since $\mathcal{L}^{t}$ can be
viewed as a summation of $\mathcal{L}_{x}^{t}$ over a set of points $x\in
\mathcal{S}^{2}$. More rigorously, we should assume that $f_{l}(x)\neq t $;
it is immediate to see that this is satisfied almost surely.

\section{Proof of Theorem \ref{thm:var=1/l^2}}

\label{proof}

The first Lemma is probably already known, but we failed to locate a direct
reference and thus we report it for completeness. The proof is based on
standard properties of Gaussian random variables.

\begin{lemma}
\label{lem:var(Dl)=*int(arcsin)} For $l$ even we have
\begin{equation}
\var(\mathcal{D}_{l})=32\pi \int\limits_{0}^{\pi /2}\arcsin
(P_{l}(\cos \theta ))\sin {\theta }d\theta .  \label{eq:var(Dl)=*int(arcsin)}
\end{equation}
\end{lemma}

We postpone the proof of Lemma \ref{lem:var(Dl)=*int(arcsin)} until the end
of the present section. Let us denote the integral in %
\eqref{eq:var(Dl)=*int(arcsin)}
\begin{equation}
I_{l}=\int\limits_{0}^{\pi /2}\arcsin (P_{l}(\cos \theta ))\sin {\theta }%
d\theta ,  \label{eq:Il def}
\end{equation}%
so that evaluating the defect variance is equivalent to evaluating $I_{l}$,
which is done in the following proposition (for even $l$):

\begin{proposition}
\label{prop:Il sim C1/l^2} As $l\rightarrow \infty $ along even numbers, we
have
\begin{equation}
I_{l}=\frac{C_{1}}{l^{2}}+o_{l\rightarrow \infty }\left( \frac{1}{l^{2}}%
\right) ,  \label{eq:Il=C1/l^2+o}
\end{equation}%
where
\begin{equation}
C_{1}=\int\limits_{0}^{\infty }\psi \left( \arcsin (J_{0}(\psi ))-J_{0}(\psi
)\right) d\psi .  \label{eq:C1 def}
\end{equation}%
Moreover, the constant $C_{1}$ is strictly positive.
\end{proposition}

\begin{proof}[Proof of Theorem \ref{thm:var=1/l^2} assuming Proposition \ref%
{prop:Il sim C1/l^2}]
Formula \eqref{eq:var(Dl)=*int(arcsin)} together with Proposition \ref%
{prop:Il sim C1/l^2} yields
\begin{equation*}
\var(\mathcal{D}_{l})=\frac{C}{l^{2}}\left( 1+o_{l\rightarrow \infty
}(1)\right) .
\end{equation*}%
The positivity of the constant $C=32\pi C_{1}$ follows directly from the
positivity of $C_{1}$, which is claimed in Proposition \ref{prop:Il sim
C1/l^2}.
\end{proof}

\begin{proof}[Proof of Proposition \ref{prop:Il sim C1/l^2}]
As it was explained in Section \ref{sec:on the proof}, to extract the
asymptotics of $I_{l}$, we will expand the arcsine on the RHS of %
\eqref{eq:Il def} into the Taylor series around the origin; we will
encounter only the odd moments of $P_{l}(\cos \theta )$, due to the arcsine
being an odd function. As it was pointed out in Section \ref{sec:on the
proof}, the function $P_{l}(\cos \theta )$ differ from its scaling limit in
that the integral of the former vanishes, whereas the integral of the latter
diverges; all the other odd moments of $P_{l}(\cos\theta)$ are asymptotic to
the (properly scaled) corresponding moments of the Bessel function, (see
Lemma \ref{lem:int Pl^k =ck/l^2}). To account for this discrepancy, we
subtract $P_{l}\left( \cos \theta \right) $ from the integrand and add it
back separately. To this end we write
\begin{equation*}
\arcsin (P_{l}(\cos \theta ))=\left( \arcsin (P_{l}(\cos \theta
))-P_{l}(\cos \theta )\right) +P_{l}(\cos \theta ),
\end{equation*}%
so that the vanishing of the integral of the latter
\begin{equation*}
\int\limits_{0}^{\pi /2}P_{l}(\cos \theta )\sin {\theta }d\theta =0,
\end{equation*}%
for even $l$ implies
\begin{equation}
I_{l}=\int\limits_{0}^{\pi /2}\left( \arcsin (P_{l}(\cos \theta
))-P_{l}(\cos \theta )\right) \sin {\theta }d\theta .
\label{eq:Il add subtract Pl}
\end{equation}%
The advantage of the latter representation \eqref{eq:Il add subtract Pl}
over \eqref{eq:Il def} is that the only powers that will appear in the
Taylor expansion of the arcsine on the right-hand side of \eqref{eq:Il add
subtract Pl} are of order $\ge 3$, so that the moments of $%
P_{l}(\cos \theta )$ are all identical to the corresponding moments of the
scaling limit. Intuitively, this means that we may replace the appearances
of $P_{l}(\cos \theta)$ in \eqref{eq:Il add subtract Pl} by its the scaling
limit. The rest of the present proof is a rigorous argument that will establish the
latter statement.

Let
\begin{equation}
\arcsin (t)-t=\sum\limits_{k=1}^{\infty }a_{k}t^{2k+1},
\label{eq:arcsin(t)-t Tay}
\end{equation}%
where
\begin{equation}
a_{k}=\frac{(2k)!}{4^{k}(k!)^{2}(2k+1)}  \label{eq:ak arcsine def}
\end{equation}%
are the Taylor coefficients of the arcsine. Note that all the terms in the
expansion \eqref{eq:arcsin(t)-t Tay} are positive, and by the Stirling
formula, the coefficients are asymptotic to
\begin{equation}
a_{k}\sim \frac{c}{k^{3/2}}  \label{eq:ak = 1/k^3/2}
\end{equation}%
for some $c>0$, so that, in particular, the Taylor series %
\eqref{eq:arcsin(t)-t Tay} is uniformly absolutely convergent. Therefore we
may write
\begin{equation}
I_{l}=\sum\limits_{k=1}^{\infty }a_{k}\int\limits_{0}^{\pi /2}P_{l}(\cos
\theta )^{2k+1}\sin {\theta }d\theta .  \label{eq:Il=sum ak int Pl^2k+1}
\end{equation}

We know from Lemma \ref{lem:int Pl^k =ck/l^2}, that for every $k\geq 1$,
\begin{equation}
\int\limits_{0}^{\pi /2}P_{l}(\cos \theta )^{2k+1}\sin {\theta }d\theta \sim
\frac{c_{2k+1}}{l^{2}},  \label{eq:Pl mom 2k+1 sim ck/l^2}
\end{equation}%
with $c_{2k+1}$ given by \eqref{eq:ck def}; comparing this result to %
\eqref{eq:Il=sum ak int Pl^2k+1}, it is natural to expect that
\begin{equation}
I_{l}=\frac{C_{2}}{l^{2}}+o\left( \frac{1}{l^{2}}\right) ,
\label{eq:Il = C2/l^2+o}
\end{equation}%
where
\begin{equation}
C_{2}=\sum\limits_{k=1}^{\infty }a_{k}c_{2k+1}.  \label{eq:C2 def}
\end{equation}

We are going to formally prove \eqref{eq:Il = C2/l^2+o} immediately;
however, first we evaluate the constant $C_{2}$ by summing up the series in %
\eqref{eq:C2 def}, and validate that indeed $C_{2}=C_{1}$ in
\eqref{eq:C1
def}, as claimed. We plug \eqref{eq:ck def} into \eqref{eq:C2 def} to
formally compute
\begin{equation*}
\begin{split}
C_{2} = &\sum\limits_{k=1}^{\infty} a_{k} \int\limits_{0}^{\infty} \psi
J_{0}(\psi)^{2k+1} d\psi = \int\limits_{0}^{\infty} \psi \cdot \left(
\sum\limits_{k=1}^{\infty} a_{k} J_{0}(\psi)^{2k+1}\right) d\psi \\
&= \int\limits_{0}^{\infty} \psi \cdot\left(
\arcsin(J_{0}(\psi))-J_{0}(\psi) \right) d\psi = C_{1},
\end{split}%
\end{equation*}
where to obtain the third equality, we used \eqref{eq:arcsin(t)-t Tay}
again. To justify the exchange of the summation and integration order, we
consider the finite summation
\begin{equation*}
\sum\limits_{k=1}^{m}a_{k}\int\limits_{0}^{\infty} \psi J_{0}(\psi)^{2k+1}
d\psi,
\end{equation*}
using \eqref{eq:ak = 1/k^3/2} and \eqref{eq:J(psi)<<1/sqrt(psi)} to bound
the contribution of tails, and take the limit $m\rightarrow\infty$.

We now turn to prove \eqref{eq:Il = C2/l^2+o}. To this end we expand %
\eqref{eq:arcsin(t)-t Tay} into a \emph{finite} degree Taylor polynomial
while controlling the tail using Lemma \ref{lem:int Pl^5 << 1/l^2}. Indeed,
using \eqref{eq:ak = 1/k^3/2} and \eqref{eq:int Pl^k << 1/l^2}, we easily
obtain
\begin{equation*}
\begin{split}
\sum_{k=m+1}^{\infty} a_{k}\int\limits_{0}^{\pi/2} \left|
P_{l}(\cos\theta)\right|^{2k+1} \sin{\theta}d\theta &\le
\sum_{k=m+1}^{\infty} a_{k}\int\limits_{0}^{\pi/2} \left|
P_{l}(\cos\theta)\right|^{5} \sin{\theta}d\theta \\
&\ll \frac{1}{l^2} \cdot \sum\limits_{k=m+1}\frac{1}{k^{3/2}} \ll \frac{1}{%
\sqrt{m}l^{2}},
\end{split}%
\end{equation*}
since $|P_{l}(t)| \le 1$ for every $l$ and $t\in [-1,1]$, so that $%
|P_{l}(t)|^{k}$ is monotone decreasing with $k$.

We then have for every $m$ ($m=m(l)$ to be chosen)
\begin{equation}  \label{eq:Il=sum ak int fin}
I_{l} = \sum\limits_{k=1}^{m}a_{k} \int\limits_{0}^{\pi/2}
P_{l}(\cos\theta)^{2k+1}\sin{\theta}d\theta + O\left( \frac{1}{\sqrt{m} l^{2}%
} \right),
\end{equation}
and plugging \eqref{eq:Pl mom 2k+1 sim ck/l^2} (a direct consequence of
Lemma \ref{lem:int Pl^k =ck/l^2}) into \eqref{eq:Il=sum ak int fin} finally
yields
\begin{equation}  \label{eq:Il=C2,m/l^2+o}
I_{l} = C_{2,m} \cdot \frac{1}{l^2} + o_{m}\left( \frac{1}{l^2} \right) +
\frac{1}{\sqrt{m}l^{2}}
\end{equation}
with
\begin{equation*}
C_{2,m} = \sum\limits_{k=1}^{m}a_{k}c_{2k+1}.
\end{equation*}
It is clear that \eqref{eq:Il=C2,m/l^2+o} implies \eqref{eq:Il = C2/l^2+o}
(recall the definition \eqref{eq:C2 def} of $C_{2}$ and note that as $%
m\rightarrow\infty$, $C_{2,m}\rightarrow C_{2}$), which concludes the proof
of the statement \eqref{eq:Il=C1/l^2+o} of the present proposition.

It then remains to prove the positivity of the constant $C_{1}$. While its
nonnegativity $C_{1}\ge 0$ is clear since it is the leading constant for the
integral $I_{1}$, which, up to an explicit positive constant, equals the
variance of a random variable via \eqref{eq:var(Dl)=*int(arcsin)}, the
strict positivity is less obvious. To this end, we recall that $C_{1} =
C_{2} $, the latter being given by \eqref{eq:C2 def}. Note that all the
Taylor coefficients $a_{k}$ of arcsine are positive (see
\eqref{eq:ak
arcsine def}). Hence the positivity of $C_{2}$ (and thus also, of $C_{1}$)
follows from the second statement of Lemma \ref{lem:int Pl^k =ck/l^2}, which
claims that $c_{k} \ge 0$ are nonnegative, and $c_{3} > 0$ is explicitly
given.
\end{proof}

\begin{proof}[Proof of Lemma \ref{lem:var(Dl)=*int(arcsin)}]
The result of the present lemma is an artifact of the following general fact
(see e.g. ~\cite{Rice1,Rice2}). Let $(X_{1},X_{2})$ be a $2$-variate centred
Gaussian random variable with covariance matrix
\begin{equation*}
\left(
\begin{matrix}
1 & r \\
r & 1%
\end{matrix}
\right),
\end{equation*}
$|r|\le 1$, and for $i=1,2$ define the random variables (recall that $H(t)$
is the Heaviside function \eqref{eq:heaviside def})
\begin{equation*}
B_{i} = H(X_{i}) =
\begin{cases}
1 & X_{i} > 0 \\
-1 & X_{i} < 0%
\end{cases}
.
\end{equation*}
Then $B_{i}$ are both mean zero with covariance given by
\begin{equation*}
Cov(B_{1}, B_{2}) = \frac{2}{\pi}\arcsin(r).
\end{equation*}

Using the definition \eqref{eq:defect def} of the defect, we may exchange
the order of taking the expectation and integrating to write
\begin{equation}
\begin{split}
\var(\mathcal{D}_{l})& =\mathbb{E}[\mathcal{D}_{l}^{2}]=\mathbb{E}%
\iint\limits_{\mathcal{S}^{2}\times \mathcal{S}%
^{2}}H(f_{l}(x))H(f_{l}(y))dxdy \\
& =\iint\limits_{\mathcal{S}^{2}\times \mathcal{S}^{2}}\mathbb{E}\left[
H(f_{l}(x))H(f_{l}(y))\right] dxdy=4\pi \int\limits_{\mathcal{S}^{2}}\mathbb{%
E}\left[ H(f_{l}(N))H(f_{l}(x))\right] dx,
\end{split}
\label{eq:defect double int}
\end{equation}%
by the isotropic property of the random field $f_{l}$, where $N$ is the
northern pole. Note that for every $x\in \mathcal{S}^{2}$, $f_{l}(N)$ and $%
f_{l}(x)$ are jointly Gaussian, centred, with unit variance and covariance
equal to
\begin{equation*}
\mathbb{E}\left[ f_{l}(N)\cdot f_{l}(x)\right] =r_{l}(N,x)=P_{l}(\cos \theta
),
\end{equation*}%
where $(\theta ,\phi )$ are the spherical coordinates of $x$; the latter
follows from the definition \eqref{eq:rl coval def} of the covariance
function. Therefore, as it was explained earlier, for every $x,y\in \mathcal{%
S}^{2}$,
\begin{equation*}
\mathbb{E}\left[ H(f_{l}(x))H(f_{l}(y))\right] =\frac{2}{\pi }\arcsin
(P_{l}(\cos \theta )).
\end{equation*}%
We then evaluate the latter integral in \eqref{eq:defect double int} in the
spherical coordinates as
\begin{equation*}
\var(\mathcal{D}_{l})=8\pi ^{2}\int\limits_{0}^{\pi }\frac{2}{\pi }%
\arcsin (P_{l}(\cos \theta ))\sin \theta d\theta =16\pi \int\limits_{0}^{\pi
}\arcsin (P_{l}(\cos \theta ))\sin \theta d\theta ,
\end{equation*}%
which, taking into account $l$ being even (and thus $P_{l}(t)$ is also
even), is the statement \eqref{eq:var(Dl)=*int(arcsin)} of the present lemma.
\end{proof}

\section{Moments of Legendre polynomials}

\label{sec:Legendre moments}

We start by recalling a basic fact on the asymptotic behaviour of Legendre
polynomials (see for instance \cite{szego}); as usual, we shall denote by $%
J_{\nu }$ the Bessel function of the first kind.

\begin{lemma}[Hilb's asymptotics]
\label{lem:Hilb} For any $\epsilon > 0$, $C>0$ we have
\begin{equation}  \label{eq:Hilb}
P_{l}(\cos{\theta}) = \left( \frac{\theta}{\sin(\theta)}
\right)^{1/2}J_{0}((l+1/2)\theta) +\delta(\theta),
\end{equation}
where
\begin{equation}  \label{eq:delta bound}
\delta(\theta) \ll
\begin{cases}
\theta^{1/2}l^{-3/2} & \theta > \frac{C}{l} \\
\theta^2 & 0 < \theta < \frac{C}{l}%
\end{cases}%
.
\end{equation}
uniformly w.r.t. $l\ge 1$, $\theta \in [0,\pi-\epsilon].$
\end{lemma}

Note that we will use Lemma \ref{lem:Hilb} only for the range $\theta \in
\lbrack 0,\frac{\pi }{2}]$, so that we may forget about the $\epsilon $
altogether. We will also recall that (see again \cite{szego})
\begin{equation}
|J_{0}(\psi )|=O\left( \frac{1}{\sqrt{\psi }}\right) .
\label{eq:J(psi)<<1/sqrt(psi)}
\end{equation}

\begin{lemma}
\label{lem:int Pl^k =ck/l^2} Let $j \ge 5$ or $j=3$. Then
\begin{equation}  \label{eq:int Pl^k =ck/l^2}
\int\limits_{0}^{\pi/2} P_{l} (\cos\theta)^j\sin{\theta} d\theta = c_{j}%
\frac{1}{l^2}(1+o_{j}(1)),
\end{equation}
where the constants $c_{j}$ are given by
\begin{equation}  \label{eq:ck def}
c_{j} = \int\limits_{0}^{\infty} \psi J_{0}(\psi)^j d\psi,
\end{equation}
the RHS of \eqref{eq:ck def} being absolutely convergent for $j\ge 5$ and
conditionally convergent for $j=3$. Moreover, for every $j$ as above, the
constants $c_{j} \ge 0$ are nonnegative, and $c_{3}>0$ is positive, given
explicitly by
\begin{equation}  \label{eq:c3 = 2/pi*sqrt(3)}
c_{3} = \frac{2}{\pi \sqrt{3}}.
\end{equation}
\end{lemma}

\begin{proof}
By the Hilb's asymptotics, we have
\begin{equation}
\int\limits_{0}^{\pi /2}P_{l}(\cos \theta )^{j}\sin {\theta }d\theta
=\int\limits_{0}^{\pi /2}\left( \left( \frac{\theta }{\sin (\theta )}\right)
^{1/2}J_{0}((l+1/2)\theta )+\delta (\theta )\right) ^{j}\sin {\theta }%
d\theta .  \label{eq:mom k = J + delta}
\end{equation}%
The contribution of the error term to \eqref{eq:mom k = J + delta} is
(exploiting $\frac{\theta }{\sin \theta }$ being bounded)
\begin{equation}
\ll \int\limits_{0}^{\pi /2}|J_{0}((l+1/2)\theta )|^{j-1}\delta (\theta
)\theta d\theta =\int\limits_{0}^{1/l}+\int\limits_{1/l}^{\pi /2}.
\label{eq:int delta 0->1/l->pi/2}
\end{equation}%
Now
\begin{equation*}
\int\limits_{0}^{1/l}|J_{0}((l+1/2)\theta )|^{j-1}\delta (\theta )\theta
d\theta \ll \int\limits_{0}^{1/l}\theta ^{3}d\theta \ll \frac{1}{l^{4}},
\end{equation*}%
and using \eqref{eq:J(psi)<<1/sqrt(psi)}, we may bound the second integral
in \eqref{eq:int delta 0->1/l->pi/2} as
\begin{equation*}
\int\limits_{1/l}^{\pi /2}|J_{0}((l+1/2)\theta )|^{j-1}\delta (\theta
)\theta d\theta \ll \frac{1}{l^{(j+2)/2}}\int\limits_{1/l}^{\pi /2}\theta
^{3/2}d\theta \ll \frac{1}{l^{(j+2)/2}}.
\end{equation*}%
Plugging the last couple of estimates into \eqref{eq:int delta 0->1/l->pi/2}%
, and finally into \eqref{eq:mom k = J + delta}, we obtain for $j\geq 3$
\begin{equation}
\int\limits_{0}^{\pi /2}P_{l}(\cos \theta )^{j}\sin {\theta }d\theta
=\int\limits_{0}^{\pi /2}\left( \frac{\theta }{\sin (\theta )}\right)
^{j/2-1}J_{0}((l+1/2)\theta )^{j}\theta d\theta +O\left( \frac{1}{l^{5/2}}%
\right) .  \label{eq:mom k <-> J}
\end{equation}

Therefore we are to evaluate
\begin{equation*}
\int\limits_{0}^{\pi /2}\left( \frac{\theta }{\sin (\theta )}\right)
^{j/2-1}J_{0}((l+1/2)\theta )^{j}\theta d\theta =\frac{1}{L^{2}}%
\int\limits_{0}^{L\pi /2}\left( \frac{\psi /L}{\sin (\psi /L)}\right)
^{j/2-1}J_{0}(\psi )^{j}\psi d\psi ,
\end{equation*}%
where we denote $L:=l+\frac{1}{2}$ for brevity. The statement of the present
lemma is then equivalent to
\begin{equation}
\int\limits_{0}^{L\pi /2}\left( \frac{\psi /L}{\sin (\psi /L)}\right)
^{j/2-1}J_{0}(\psi )^{j}\psi d\psi \rightarrow c_{j},
\label{eq:int sinx/x 1}
\end{equation}%
where $c_{j}$ is defined by \eqref{eq:ck def}.

The main idea is that the leading contribution is provided from the
intermediate range $1\ll \psi \ll \epsilon L$ for any $\epsilon >0$; here we
may replace the factor
\begin{equation*}
\left( \frac{\psi /L}{\sin (\psi /L)}\right) ^{j/2-1}
\end{equation*}%
with $1$. To make this argument precise, we write for $\psi \in \lbrack 0,%
\frac{\pi }{2}\cdot L]$
\begin{equation*}
\frac{\psi /L}{\sin (\psi /L)}=1+O\left( \frac{\psi ^{2}}{L^{2}}\right) ,
\end{equation*}%
so that also
\begin{equation}
\left( \frac{\psi /L}{\sin (\psi /L)}\right) ^{j/2-1}=1+O_{j}\left( \frac{%
\psi ^{2}}{L^{2}}\right) .  \label{eq:sinx/x=1+O(x^2)}
\end{equation}%
Therefore the integral in \eqref{eq:int sinx/x 1} is
\begin{equation}
\int\limits_{0}^{L\pi /2}\left( \frac{\psi /L}{\sin (\psi /L)}\right)
^{j/2-1}J_{0}(\psi )^{j}\psi d\psi =\int\limits_{0}^{L\pi /2}J_{0}(\psi
)^{j}\psi d\psi +O\left( \frac{1}{L^{2}}\int\limits_{0}^{L\pi /2}\psi
^{3}|J_{0}(\psi )|^{j}d\psi \right) .  \label{eq:rep sinx/x -> 1}
\end{equation}

Note that as $l\rightarrow \infty $ (equivalently $L\rightarrow \infty $),
the main term on the RHS of converges to $c_{j};$ indeed
\begin{equation}
\int\limits_{0}^{L\pi /2}J_{0}(\psi )^{j}\psi d\psi \rightarrow c_{j},
\label{eq:int J^kpsi -> ck}
\end{equation}%
so that it remains to bound the error term. Now
\begin{equation}
\int\limits_{0}^{L\pi /2}\psi ^{3}|J_{0}(\psi )^{j}|d\psi
=\int_{0}^{1}+\int_{1}^{L\pi /2}=O(1)+\int_{1}^{L\pi /2},
\label{eq:int 0->L psi^3J^k 0->1->L}
\end{equation}%
and we use \eqref{eq:J(psi)<<1/sqrt(psi)} to bound the latter as
\begin{equation*}
\int_{1}^{L\pi /2}\psi ^{3}|J_{0}(\psi )^{j}|d\psi \ll \int\limits_{1}^{L\pi
/2}\psi ^{3-j/2}d\psi =O(1+l^{4-j/2}),
\end{equation*}%
so that upon plugging the latter into \eqref{eq:int 0->L psi^3J^k 0->1->L}
yields
\begin{equation*}
\int\limits_{0}^{L\pi /2}\psi ^{3}|J_{0}(\psi )^{j}|d\psi =O(1+l^{4-j/2});
\end{equation*}%
plugging the latter into \eqref{eq:rep sinx/x -> 1} yields
\begin{equation*}
\int\limits_{0}^{L\pi /2}\left( \frac{\psi /L}{\sin (\psi /L)}\right)
^{j/2-1}J_{0}(\psi )^{j}\psi d\psi =\int\limits_{0}^{L\pi /2}J_{0}(\psi
)^{j}\psi d\psi +O(l^{-2}+l^{2-j/2}),
\end{equation*}%
so that \eqref{eq:int J^kpsi -> ck} implies \eqref{eq:int sinx/x 1} for $%
j\geq 5$, which was equivalent to the statement of the present lemma in this
case.

It then remains to prove the result for $j=3$, for which case we have to
work a little harder due to the conditional convergence of the integral %
\eqref{eq:ck def}; to treat this technicality we will have to exploit the
oscillatory behaviour of the Bessel function, and not only its decay %
\eqref{eq:J(psi)<<1/sqrt(psi)}. It is well-known that
\begin{equation*}
J_{0}(\psi )=\sqrt{\frac{2}{\pi }}\frac{\cos (\psi -\pi /4)}{\sqrt{\psi }}%
+O\left( \frac{1}{\psi ^{3/2}}\right) ,
\end{equation*}%
so that
\begin{equation}
\left( \frac{2}{\pi }\right) ^{3/2}\frac{\cos (\psi -\pi /4)^{3}}{\psi ^{3/2}%
}+O\left( \frac{1}{\psi ^{5/2}}\right) ,  \label{eq:J^3 asymp}
\end{equation}%
and hence the integral on the RHS \eqref{eq:ck def} is indeed convergent, by
integration by parts.

Now we choose a (large) parameter $K\gg 1$, divide the integration range
into $[0,K]$ and $[K,L\frac{\pi }{2}]$; the main contribution comes from the
first term, whence we need to prove that the latter vanishes. Indeed, we use %
\eqref{eq:J^3 asymp} to bound
\begin{equation}
\int\limits_{K}^{L\pi /2}\left( \frac{\psi /L}{\sin (\psi /L)}\right)
^{1/2}J_{0}(\psi )^{3}\psi d\psi \ll \frac{1}{\sqrt{K}},
\label{eq:intJ^3 C->L =O(1/sqrt{C})}
\end{equation}%
where we use integration by parts with the bounded function
\begin{equation*}
I(T)=\int\limits_{0}^{T}\cos (t)^{3}dt
\end{equation*}%
to bound the contribution leading term in \eqref{eq:J^3 asymp}; it is easy
to bound the contribution of the error term in \eqref{eq:J^3 asymp} using
the crude estimate \eqref{eq:J(psi)<<1/sqrt(psi)}.

On $[0,K]$ we use \eqref{eq:sinx/x=1+O(x^2)} to write
\begin{equation}
\int\limits_{0}^{K}\left( \frac{\psi /L}{\sin (\psi /L)}\right)
^{1/2}J_{0}(\psi )^{3}\psi d\psi =\int\limits_{0}^{K}J_{0}(\psi )^{3}\psi
d\psi +O\left( \frac{K^{5/2}}{l^{2}}\right) ,  \label{eq:intJ^3 0->C}
\end{equation}%
the former clearly being convergent to $c_{3}$. Combining
\eqref{eq:intJ^3
C->L =O(1/sqrt{C})} with \eqref{eq:intJ^3 0->C} we obtain
\begin{equation*}
\int\limits_{K}^{L\pi /2}\left( \frac{\psi /L}{\sin (\psi /L)}\right)
^{1/2}J_{0}(\psi )^{3}\psi d\psi =\int\limits_{0}^{K}J_{0}(\psi )^{3}\psi
d\psi +O\left( \frac{1}{\sqrt{K}}+\frac{K^{5/2}}{l^{2}}\right) ;
\end{equation*}%
this implies \eqref{eq:int sinx/x 1} for $j=3$ (which is equivalent to the
statement of the present lemma for $j=3$) upon choosing the parameter $K$
growing to infinity sufficiently slowly (e.g. $K=\sqrt{l}$). Since $j=3$ was
the only case that was not covered earlier, this concludes proof of the
moments part of Lemma \ref{lem:int Pl^k =ck/l^2}.

It then remains to prove the nonnegativity statement of $c_{j}$, and the
explicit expression \eqref{eq:c3 = 2/pi*sqrt(3)} for $c_{3}$. In fact, both
of those statements follow from the computation we performed in our previous
paper ~\cite{MaWi}, p. 18. We report the relevant results in Lemma \ref%
{MaWiProp} below.
\end{proof}

We recall an alternative characterization for moments of Legendre
polynomials from ~\cite{MaWi}. For given positive integers $l_{1},l_{2},l_{3}
$ we introduce the so-called Clebsch-Gordan coefficients $\left\{
C_{l_{1}0l_{2}0}^{l_{3}0}\right\} ,$ which are different from zero if and
only if $l_{1},l_{2},l_{3}$ are such that $l_{1}+l_{2}+l_{3}$ is even and $%
l_{i}+l_{j}\leq l_{k}$ for all permutations $i,j,k=1,2,3.$ The
Clebsch-Gordan coefficients are well-known in group representation theory
(they intertwine alternative representations for $SO(3)$) and in the quantum
theory of angular momentum; we do not provide more details here, but we
refer instead to \cite{MaWi} or to standard references such as \cite{VMK,VIK}
(see also \cite{MaPeCUP}). The results provided in \cite{MaWi} are as
follows:

\begin{lemma}[~\protect\cite{MaWi}, Lemma A.1 and above]
\label{MaWiProp}For all even $l$, we have
\begin{equation*}
\int\limits_{0}^{1}P_{l}(t) ^{3}dt=\frac{1}{2l+1}\left\{
C_{l0l0}^{l0}\right\} ^{2}\text{ ,}
\end{equation*}%
and
\begin{equation*}
\lim_{l\rightarrow \infty }\frac{1}{2l+1}\left\{ C_{l0l0}^{l0}\right\} ^{2}=%
\frac{2}{\pi \sqrt{3}}.
\end{equation*}%
Also, for $j\geq 5$
\begin{equation*}
\int\limits_{0}^{1}P_{l}(t)^{j}dt=\frac{1}{2l+1}\sum_{L_{1}...L_{j-3}}\left%
\{ C_{l0l0}^{L_{1}0}C_{L_{1}0l0}^{L_{2}0}...C_{L_{j-3}0l0}^{l0}\right\}
^{2}>0\text{ .}
\end{equation*}
\end{lemma}

\begin{verbatim}

\end{verbatim}

\begin{remark}
The previous discussion yields the following interesting corollary: as $%
l\rightarrow \infty $%
\begin{equation*}
\lim_{l\rightarrow \infty }\frac{1}{2l+1}\left\{ C_{l0l0}^{l0}\right\} ^{2}
=\int\limits_{0}^{\infty }\psi J_{0}(\psi )^{3}d\psi \text{,}
\end{equation*}
and for $j\geq5$
\begin{equation*}
\lim_{l\rightarrow \infty }\left[ \frac{1}{2l+1}\sum_{L_{1}...L_{j-3}}\left%
\{ C_{l0l0}^{L_{1}0}C_{L_{1}0l0}^{L_{2}0}...C_{L_{j-3}0l0}^{l0}\right\} ^{2}%
\right] =\int\limits_{0}^{\infty }\psi J_{0}(\psi )^{j}d\psi \text{.}
\end{equation*}%
\bigskip
\end{remark}

The following lemma establishes the lower bound for the constant $C$ in
Theorem \ref{thm:var=1/l^2} claimed in Section \ref{sec:main res}. Its proof
is straightforward.

\begin{lemma}
\label{lem:C>2/pi^2 /sqrt(27)} For $C$ as in \eqref{eq:var=1/l^2}, we have
\begin{equation*}
C>\frac{32}{\sqrt{27}}.
\end{equation*}
\end{lemma}

\begin{proof}
Recall that $C=\frac{2}{\pi }C_{1}$, where $C_{1}=C_{2}$ is given by %
\eqref{eq:C2 def}. Since all the $a_{j}$ and $c_{k}$ are nonnegative, Lemma %
\ref{lem:C>2/pi^2 /sqrt(27)} follows from bounding the first term in the
series
\begin{equation*}
C\geq 32\pi \times a_{1}c_{3}=\frac{32}{\sqrt{27}},
\end{equation*}%
where we used $a_{1}=1/6$ and \eqref{eq:c3 = 2/pi*sqrt(3)}.
\end{proof}

\begin{remark}
{\em
As a final remark, we note that the moments of $P_{l}(\cos {\theta })$ are
themselves of some physical interest. In particular, while analyzing the
relationship between asymptotic Gaussianity and ergodicity of isotropic
spherical random fields, \cite{MaPeJMP} considered the random field
\begin{equation*}
T(x):=\sum_{l=1}^{\infty }T_{l}(x)=\sum_{l=1}^{\infty }c_{l}Y_{l0}(g\cdot x),
\end{equation*}%
where $Y_{l0}(\theta ,\phi )=\sqrt{\frac{2l+1}{4\pi }}P_{l}(\cos \theta )$
denotes standard spherical harmonics (for $m=0)$ and $g\in SO(3)$ is a
uniformly distributed random rotation in $\mathbb{R}^{3}.$ It is readily
seen that the resulting field is isotropic, and each of its Fourier
components at frequency $l$ is marginally distributed as (for any $x\in
\mathcal{S}^{2})$%
\begin{equation*}
Y_{l}=c_{l}\sqrt{\frac{2l+1}{4\pi }}P_{l}(t)\text{, }t\sim U[0,1].
\end{equation*}%
Up to normalization constants, then, the moments of the marginal law are
exactly those of the Legendre polynomials we established earlier in the
paper. More precisely, if we focus (as in \cite{MaPeJMP}) on $\widetilde{T}%
_{l}:=T_{l}(x)/\sqrt{Var(T_{l})},$ we obtain immediately, for any $x\in
\mathcal{S}^{2}$%
\begin{equation*}
\mathbb{E}\left\{ \widetilde{T}_{l}(x)\right\} =\mathbb{E}\left\{ \sqrt{%
(2l+1)}P_{l}(t)\right\} =0,\text{ }\mathbb{E}\left\{ \widetilde{T}%
_{l}(x)\right\} ^{2}=1,
\end{equation*}%
and%
\begin{equation*}
\mathbb{E}\left\{ \widetilde{T}_{l}(x)\right\} ^{3}=E\left\{ \sqrt{(2l+1)}%
P_{l}(t)\right\} ^{3}=O\left(\frac{1}{\sqrt{l}}\right)\text{,}
\end{equation*}
whereas ~\cite{Wig2}
\begin{equation*}
E\left\{ \sqrt{(2l+1)}P_{l}(t)\right\} ^{4}\simeq \log l\text{ ,}
\end{equation*}%
and hence all moments of order $q\geq 4$ diverge.}

\end{remark}

Note that in Lemma \ref{lem:int Pl^k =ck/l^2} we worked relatively hard to
establish the precise asymptotics for the moments of Legendre polynomials. A
much cruder version of the same argument gives a uniform upper bound for
the $5$th moment of the \emph{absolute value} of the Legendre polynomials:

\begin{lemma}
\label{lem:int Pl^5 << 1/l^2} We have the following uniform upper bound for
the $5$th moment of the absolute value of Legendre polynomials
\begin{equation}
\int\limits_{0}^{\pi /2}|P_{l}(\cos {\theta })|^{5}\sin \theta d\theta=O\left(
\frac{1}{l^{2}}\right) ,  \label{eq:int Pl^k << 1/l^2}
\end{equation}%
where the constant involved in the $`O^{\prime }$-notation is universal.
\end{lemma}

\begin{proof}
The proof follows along the same lines as the beginning of proof of Lemma %
\ref{lem:int Pl^k =ck/l^2} for $k\geq 5$, except that we use the crude upper
bound \eqref{eq:J(psi)<<1/sqrt(psi)} whenever we reach \eqref{eq:mom k <-> J}%
, and the trivial inequality%
\begin{equation*}
\frac{\theta }{\sin {\theta }}\ll 1,\text{ for }\theta \in \left[ 0,\frac{%
\pi }{2}\right] .
\end{equation*}
\end{proof}

\end{document}